\pdfoutput=1
\documentclass{article}
\usepackage{graphicx}
\usepackage{cite}
\usepackage{picinpar}
\usepackage{amsmath}
\usepackage{url}
\usepackage{flushend}
\usepackage[latin1]{inputenc}
\usepackage{colortbl}
\usepackage{soul}
\usepackage{multirow}
\usepackage{pifont}
\usepackage{color}
\usepackage{alltt}
\usepackage[hidelinks]{hyperref}
\usepackage{enumerate}
\usepackage{siunitx}
\usepackage{breakurl}
\usepackage{epstopdf}
\usepackage{pbox}
\usepackage{subfigure}
\usepackage{amssymb}\usepackage{indentfirst}
\usepackage[thmmarks,amsmath]{ntheorem}
\usepackage{booktabs}
\begin{document}
\theorembodyfont{\normalfont}
\theoremheaderfont{\itshape}%¶šÀí±êÌâ×ÖÌå
\theoremseparator{:}
\newtheorem{theorem}{\indent Theorem}
\newtheorem{lemma}{\indent Lemma}
\newtheorem{assumption}{\indent Assumption}
\newtheorem{proposition}{\indent Proposition}
\newtheorem{corollary}{\indent Corollary}
\newtheorem{definition}{\indent Definition}%¶šÒå¶šÒå»·Ÿ³ÃûÎªdefinition£¬ÏÔÊŸÈç¡°¶šÀí1.1.1¡±
\newtheorem{remark}{\indent Remark}
\newtheorem{example}{\indent Example}
\newtheorem{proof}{\indent Proof}
\def\QEDclosed{\mbox{\rule[0pt]{1.5ex}{1.5ex}}} % ¶šÒåÊµÐÄ·û
\def\QED{\QEDclosed} % Ñ¡Ìî\QEDµÃµœÊµÐÄ
\def\endproof{\hspace*{\fill}~\QED\par\endtrivlist\unskip}% ÔÚproof »·Ÿ³×Ô¶¯ÌíŒÓœáÊø·û
\newcommand\blfootnote[1]{%
\begingroup
\renewcommand\thefootnote{}\footnote{#1}%
\addtocounter{footnote}{-1}%
\endgroup
}
\title{Stabilizability of multi-agent systems under event-triggered controllers}
\author{Yinshuang Sun, Zhijian Ji, Yungang Liu, and Chong Lin
\thanks{E-mail address: jizhijian@pku.org.cn. (Zhjian Ji)
This work was supported by the National Natural Science Foundation of China (Grant Nos. 61873136 and 62033007), Taishan Scholars Climbing Program of Shandong Province of China and Taishan Scholars Project of Shandong Province of China (No. ts20190930).
Yinshuang Sun, Zhijian Ji (Corresponding author) and Chong Lin are with Institute of Complexity Science, College of Automation, Qingdao University, and Shandong Key Laboratory of Industrial Control Technology. Yungang Liu is with School of Control Science and Engineering, Shandong University.}}\maketitle
\begin{abstract}
In view of the problems of large consumption of communication and computing resources in the control process, this note studies a fundamental property for a class of multi-agent systems under event-triggered strategy: the S-stabilizability of a group of multi-agent systems with general linear dynamics under weakly connected directed topology.
The results indicate that the S-stabilizability can be described in some way that the stabilizability region and feedback gain can evaluate the performance of the protocol.
Firstly, a new distributed event-triggered protocol is proposed.
Under this protocol, a kind of hybrid static and dynamic event-triggered strategy are presented, respectively.
In particular, by using Lyapunov stability theory and graph partition tool, it is proved that the proposed
event-triggered control strategy can guarantee the closed-loop system achieve S-stabilizability effectively, if at least one vertex in each iSCC cell receives information from the leader, which reflects the ability of distributed control law.
Further, we demonstrate that the stabilizability can be realized if the initial system matrix $A$ is Hurwitz.
Moreover, it is confirmed that the designed static event-triggered condition is a limit case of dynamic event condition and can guarantee Zeno-free behavior.
Finally, the validity of the theoretical results is proved by numerical simulation.
\end{abstract}
\section{Introduction}
{I}{n} recent years, the distributed cooperative control \cite{Lix} \cite{Xiao}of multi-agent system has become an issue of widespread concern in control engineering, biology, physics and other disciplines because of its strong practical application background. For example, controllability \cite{direct} --\cite{almost} and consensus \cite{Necsu} \cite{Ren}.
As an effective method to deal with formation control problems in recent years, controllability has become an essential direction of multi-agent system research.
The controllability of multi-agent system was first proposed by Tanner \cite{Tanner}.
It should be noted that the concept of controllability essentially includes the possibility of executing any position at any time, which leads to some shortcomings of multi-agent systems in corresponding application fields.
Therefore, some scholars have raised a question about whether we ask too much.
This inspired the naming of ``stabilizability" of multi-agent systems with single integrator \cite{kim}.
In recent years, the research on the stabilizability of multi-agent systems has become more and more extensive.
\cite{xianzhu} \cite{Ys} extended model to general linear multi-agent systems.
However, the above studies on the stabilizability of multi-agent systems focus on the network structure and local information feedback including neighbor state feedback and self-state feedback.

Actually, the design of neighbor state feedback gain is only to adjust the interconnect gain, which has been applied to plague control in some power grids \cite{Sun}.
An interesting theoretical question is whether the whole network can only be stabilized by adjusting the interconnect gain.
For an interconnected continuous system consisting of two linear subsystems, Duan et al. solved this problem completely \cite{Huang}, where the designed interaction gain is called harmonic control.
However, how to design the interaction gain to stabilize a complex network composed of more than two subsystems is still an open problem.
In 2018, Liu et al. studied the stabilizability of heterogeneous multi-agent systems under harmonic control \cite{xianzhuLiu}.

In practical applications, the energy of the agent itself and the bandwidth of communication channel are limited.
In general, measurement, communication and control protocol updates in control tasks are performed periodically, i.e. the periodic sample control method \cite{double} \cite{intermittent}.
In order to guarantee the performance of all execution points, the sampling time constant usually takes a conservative value, which usually results in a waste of communication and computing resources.
With the deepening of research and solving the waste of computing and communication resources caused by the periodic execution of tasks by the controller in practical application, the multi-agent control strategy based on event-trigger was proposed and had attracted wide attention \cite{ODemir} \cite{KaienLiu}.
Under the event-triggered control strategy, control tasks are only executed on demand, so as to save system resources including the computing power, communication capability, and energy reserve of the agent.
At present, many meaningful research results have been achieved on the consensus of multi-agent systems based on event-triggered control.
In \cite{DimarogonasDV}, event-triggered control protocol and condition related to the state of the agent for first-order systems were designed. In addition, it is confirmed that there is no Zeno behavior. The results showed that the event-triggered control has the expected performance and reduces the number of samples.
In \cite{Liw}, the consensus of second-order multi-agent systems under event-triggered conditions was considered.
The event-triggered control of linear multi-agent systems and nonlinear systems were investigated in \cite{FzLi}, respectively.

Under directed topology, although some literatures considered the consensus of multi-agent system under event-triggered condition, as far as the author knows, the stabilizability has hardly been discussed because system matrix $A$ exists divergence.
Moreover, the influence of topology structure on the stabilizability under event-triggered controller has not been considered fully.
Based on the above challenging problems,
this note focuses on a systematic study about stabilizability of general linear multi-agent systems under event-triggered controller in directed topology, a kind of hybrid static and dynamic event-triggered conditions based on system state are given. In addition, from the designed event-triggered conditions, it can be seen that each agent does not need to monitor the state of neighboring agents continuously, hence this reduces the frequency of driving updates and communication among agents.
Based on Lyapunov stability theory, it is proved that the proposed control strategy and designed event-triggered condition can guarantee the closed-loop system achieve stabilizability effectively, and there is no Zeno behavior.
The main contributions of this work are stated as follows. (1) The definition of S-stabilizability is proposed, and a new trigger protocol that can guarantee S-stabilizability is designed; (2) It is proved that the stabilizability is a special case of S-stabilizability. If the initial system matrix $A$ is stable, the closed-loop system can realize stabilization, that is, $\mathop {\lim }\limits_{t \to \infty } {x_i}\left( t \right) = 0$; (3) The influence of topology structure on the stabilizability of the system is revealed from the perspective of graph division; (4) The designed dynamic event-triggered condition has obvious advantages than static trigger condition in reducing the number of events.

The structure of this note is as follows:
In Section \uppercase\expandafter{\romannumeral2}, we briefly introduce the concept and problem description of graph theory.
In Section \uppercase\expandafter{\romannumeral3}, the event-triggered strategy and two kinds of event-triggered conditions are established.
the accuracy of the theorem is verified by simulation experiments in Section \uppercase\expandafter{\romannumeral4}.
Section \uppercase\expandafter{\romannumeral5} summarizes this note.

\section{Preliminaries and problem description}
\subsection{Theory of graph}
A directed graph is represented by $\mathcal{G}{\rm{ = }}\left( {\mathcal{V},\mathcal{E},\mathcal{A}} \right)$, where $\mathcal{V} = \left\{ {1, \cdots ,N} \right\}$ is the set of vertex,
$\mathcal{F} = \left\{ {1, \ldots ,m} \right\}$ and $\mathcal{R} = \left\{ m+1, \ldots , N \right\}$ represent the set of followers and leaders, respectively.
$ \mathcal{E} \subseteq \mathcal{ V} \times \mathcal{ V}$ is the edge set, and $\mathcal{A} = \left( {{a_{ij}}} \right) \in {\mathbb{R}^{N \times N}}$ denotes adjacency matrix of $\mathcal{G}$. ${a_{ij}}> 0$ represents $\left( {i,j} \right) \in \mathcal{E}$.
Assume that $\left( {i,i} \right) \notin \mathcal{E}$, hence ${a_{ii}} \ne 0$.
Note that ${a_{ij}} >  0$ indicates agent $i $ can receive information from agent $j$, in which $i$ is called child vertex and $j$ is called parent vertex.
Here we choose ${\mathcal{N}_i}{\rm{ = }}\left\{ {j \in \mathcal{V},\left( {i,j} \right) \in \mathcal{E},j \ne i} \right\}$ as the set of neighbors of vertex ${i}$ in $\mathcal{V}$.
For a given graph $\mathcal{G}$, with adjacency matrix $\mathcal{A}$, the Laplacian matrix is $\mathcal{L} =\mathcal{D} -\mathcal{A}$, where $\mathcal{D}$ is a diagonal matrix with diagonal elements ${d_{ii}} = \sum_{j \in {\mathcal{N}_i}} {{a_{ij}}}$.
Therefore, the element in $\mathcal{L}$ is \[{\mathcal{L}_{ik}} = \left\{ {\begin{array}{*{20}{c}}
{\sum\limits_{j \in {\mathcal{N}_i}} {{a_{ij}}} ,}&{k = i}\\
{ - {a_{ik}},}&{k \ne i.}
\end{array}} \right.\]
A directed path from vertex $i$ to vertex $j$ is a sequence of ordered edges in the form of
$\left( {{s},{{s + 1}}} \right) \in \mathcal{E}$, where $s = i, \ldots ,j - 1$.
A weak path refers to the existence of $\left( {{s},{{s + 1}}} \right) \in \mathcal{E}$ or $\left( {{{s+1}},{s}} \right) \in \mathcal{E}$.
The graph $\mathcal{G}$ contains a directed spanning tree, if there exists a root vertex so that there exists a directed path from the root vertex to any other vertices.
If there exists a weak path between every pair of distinct vertices, then graph $\mathcal{ G}$ is said to be weakly connected.
\subsection{Basic Definitions and Lemmas}
\begin{lemma}\cite{Ren}\label{yin1}
Consider a weighted directed topology graph $\mathcal{G}$,
$\mathcal{L}$ contains a simple zero eigenvalue, and all the non-zero eigenvalues are with positive real parts if and only if graph $\mathcal{G}$ contains a directed spanning tree.
Without losing generality, we sort the eigenvalues of $\mathcal{L}$ as
\[0 = {\lambda _1}\left( \mathcal{L} \right) \le {\mathop{\rm Re}\nolimits} \left( {{\lambda _2}\left( \mathcal{L} \right)} \right) \le  \cdots  \le {\mathop{\rm Re}\nolimits} \left( {{\lambda _N}\left( \mathcal{L} \right)} \right)\]
\end{lemma}
\begin{lemma}\cite{kim}\label{yin2}
Matrix $A = \left[ {{a_{ik}}} \right] \in {\mathbb{R}^{n \times n}}$ is strictly diagonally dominant if it satisfies:
\\(1) ${\kern 1pt} \left| {{a_{ii}}} \right| \ge \sum\limits_{k = 1,k \ne i}^n {\left| {{a_{ik}}} \right|} $ \ for all $i = 1, \cdots ,n$;
\\(2) ${\kern 1pt} {\kern 1pt} \left| {{a_{ii}}} \right| > \sum\limits_{k = 1,k \ne i}^n {\left| {{a_{ik}}} \right|} $ for at least one $i$.
\end{lemma}
\begin{definition}
[independent strongly connected component]
An independent Strongly Connected Component (iSCC) of graph $\mathcal{G}{\rm{ = }}\left( {\mathcal{V},\mathcal{E},\mathcal{A}} \right)$ is the largest strongly connected induced subgraph ${\mathord\mathcal{{\buildrel{\lower3pt\hbox{$\scriptscriptstyle\smile$}}
\over L} }_1}{\rm{ = }}\left( {{\mathord\mathcal{{\buildrel{\lower3pt\hbox{$\scriptscriptstyle\smile$}}
\over V} }_1},{\mathord\mathcal{{\buildrel{\lower3pt\hbox{$\scriptscriptstyle\smile$}}
\over E} }_1},{\mathord\mathcal{{\buildrel{\lower3pt\hbox{$\scriptscriptstyle\smile$}}
\over A} }_1}} \right)$, and for any ${v_j} \in \mathcal{V}\backslash {\mathord\mathcal{{\buildrel{\lower3pt\hbox{$\scriptscriptstyle\smile$}}
\over V} }_1},{v_i} \in {\mathord\mathcal{{\buildrel{\lower3pt\hbox{$\scriptscriptstyle\smile$}}
\over V} }_1}$ satisfies $\left( {{v_i},{v_j}} \right) \notin \mathcal{E}$.
\end{definition}
Under this partition, $\mathcal{L}$ can be written as follows:
\[\mathcal{L} = \left[ {\begin{array}{*{20}{c}}
{{\mathord\mathcal{{\buildrel{\lower3pt\hbox{$\scriptscriptstyle\smile$}}
\over L} }_{11}}}&0& \cdots &0&0\\
0&{{\mathord\mathcal{{\buildrel{\lower3pt\hbox{$\scriptscriptstyle\smile$}}
\over L} }_{22}}}& \cdots &0&0\\
 \vdots & \vdots & \ddots & \vdots & \vdots \\
0&0& \cdots &{{\mathord\mathcal{{\buildrel{\lower3pt\hbox{$\scriptscriptstyle\smile$}}
\over L} }_{cc}}}&0\\
{{\mathcal{L}_{c + 1,1}}}&{{\mathcal{L}_{c + 1,2}}}& \cdots &{{\mathcal{L}_{c + 1,c}}}&{{\mathord\mathcal{{\buildrel{\lower3pt\hbox{$\scriptscriptstyle\smile$}}
\over L} }_{c + 1,c + 1}}}
\end{array}} \right]\]
By Lemmas \ref{yin1}, \ref{yin2}, ${\mathord\mathcal{\buildrel{\lower3pt\hbox{$\scriptscriptstyle\smile$}}
\over L} _{jj}} \in {\mathbb{R}^{\left| {{N_j}} \right| \times \left| {{N_j}} \right|}}\left( {j = 1, \cdots ,c} \right)$ has a zero eigenvalue, and the other non-zero eigenvalues have positive real parts;
${{\mathord\mathcal{{\buildrel{\lower3pt\hbox{$\scriptscriptstyle\smile$}}
\over L} }_{c + 1,c + 1}}}$ is Hurwitz matrix.
\begin{lemma}\cite{xianzhu}
Suppose that directed topology $\mathcal{G}$ is weakly connected and $\mathcal{L}$ is the Laplacian matrix of $\mathcal{G}$. Then $rank\left( \mathcal{L} \right) = N - c$ holds if and only if $\mathcal{G}$ contains $c$ iSCC cells.
\end{lemma}
\begin{definition}
Denote $X = \left\{ {{x_{m + 1}}, \ldots ,{x_N}} \right\}$ as the status of leaders. Then the convex hull containing all the points of $X$ can be described as
\[S: = \left\{ {\sum\nolimits_{j = m + 1}^N {{t_j}{x_j}} |{x_j} \in X,\sum\nolimits_{j = m + 1}^N {{t_j}}  = 1,{t_j} \in \left[ {0,1} \right]} \right\}\]
\end{definition}
\begin{lemma}\cite{CNowzari}\label{lemmays}
[Young's inequality]
Given $x,y \in {R^n}$, for $\upsilon  > 0$, $2{x^T}y \le \upsilon {x^T}x + \frac{{y^T}y}{\upsilon }$.
\end{lemma}
\begin{remark}
Given a directed topology $\mathcal{G}$, it is assumed that follower network $\mathcal{G}_\mathcal{F}$ is weakly connected and contains $c\left( {c \ge 1} \right)$ distinct iSCC cells.
Matrix $H = {\mathcal{L}_\mathcal{F}} + \sum\nolimits_{j = m + 1}^N {{D_j}} $ is strictly diagonally dominant if for each iSCC cell, there exists at least one vertex $i\left(i \in \mathcal{F} \right)$ such that ${d_{ij}} > 0\left(j \in \mathcal{R}\right)$.
And all eigenvalues lie in the open right-half complex-plane ${C_{ > 0}}$, where ${D_j} = diag\left\{ {{d_{1j}}, \cdots ,{d_{mj}}} \right\}$, and ${\mathcal{L}_\mathcal{F}}$ is the Laplacian matrix corresponding to the follower network $\mathcal{G}_\mathcal{F}$.
\end{remark}
\begin{lemma}\cite{ZLi}\label{lemma3}
[Comparison principle]
Consider a differential equation $\frac{{du}}{{dt}} = f\left( {t,u} \right),u\left( {{t_0}} \right) = {u_0}$, where $t > 0$, $f\left( {t,u} \right)$ is continuous and satisfies the local Lipschitz condition in $t$. Let $\left[ {{t_0},T} \right)$ be the maximum existence interval of the solution $u$, where $T$ can be infinite. If, for any $t \in \left[ {{t_0},T} \right),v = v\left( t \right)$ satisfies
\[\frac{{dv}}{{dt}} \le f\left( {t,v} \right),v\left( {{t_0}} \right) \le {u_0},\] then $v\left( t \right) \le u\left( t \right),t \in \left[ {{t_0},T} \right)$.
\end{lemma}
\subsection{Problem description}
Consider a leader-follower network consisting of $N$ agents, the dynamics of agent $i$ is described as
\begin{equation}\label{xtfc1}
\begin{split}
\left\{ {\begin{array}{*{20}{c}}
{{{\dot x}_i}\left( t \right) = A{x_i}\left( t \right),}&{i \in {\cal R}}\\
{{{\dot x}_i}\left( t \right) = A{x_i}\left( t \right) + B{u_i}\left( t \right),}&{i \in {\cal F}}
\end{array}} \right.
\end{split}
\end{equation}
where ${x_i}\left( t \right) \in \mathbb{R}^n$ and ${u_i}\left( t \right)\in \mathbb{R}^p$ represent the state and inter-agents control input of agent $i$, respectively.
$A\in \mathbb{R}^{n \times n}$ is the system matrix and $B\in \mathbb{R}^{n \times p}$ is the input matrix.
$\mathcal{F} = \left\{ {1, \ldots ,m} \right\}$ represents the set of followers,
$\mathcal{R} = \left\{ m+1, \ldots , N \right\}$ represents the set of leaders.
\begin{remark}
For undirected graphs, multi-agent systems \eqref{xtfc1} can only achieve consensus under protocol
${u_i}\left( t \right) = K\sum\nolimits_{j \in {N_i}} {{a_{ij}}\left( {{x_j} - {x_i}} \right)} $.
Under a directed topology, to the author's knowledge, there is almost no literature that can guarantee the system state converges to 0.
 Therefore, we relax the objective. Accordingly, we propose the definition of S-stabilizability below.
\end{remark}
\begin{definition}
[S-stabilizability]
The S-stabilizability of a networked system \eqref{xtfc1} can be realized, if for any initial state $x\left( 0 \right) = \left[ x_1^T\left( 0 \right), \ldots ,x_m^T\left( 0 \right) \right]^T$, the system state $x\left(t \right)$ can be driven by the control input $u\left(t \right)$ into the convex hull $S$ formed by the leader state.
\end{definition}
% or

\section{The main results}
\subsection{Static Event-Triggered Control (SETC)}
In this subsection, we solve the S-stabilizability of linear multi-agent systems under event-triggered conditions. To achieve this,
we explicitly make several key Assumptions before analysis.
\begin{assumption}\label{gf}
The follower network $\mathcal{G}_\mathcal{F}$ is weakly connected and contains $c$ distinct iSCC cells.
\end{assumption}
\begin{assumption}\label{js3}
The pair $\left( {A,B} \right)$ is stabilizable.
\end{assumption}
It should be noted that all eigenvalues of $A$ cannot be guaranteed to be on the closed left half-complex plane.

Under Assumption \ref{js3}, there is a symmetric positive definite matrix $R > 0$ that satisfies Riccati inequality with $\varsigma >0$,
\[
{A^T}R{\rm{ + }}RA - \varsigma RB{B^T}R <  - \varsigma I.
\]
Next, the event-triggered control strategy will be designed.
For agent $i$, define ${\left\{ {t_k^i} \right\}_{k \in {\mathbb{Z}_{ \ge 0}}}}$ as the trigger instant sequence with $t_{k + 1}^i = \inf \left\{ {t:{f_i}\left( t \right) \ge 0,t > t_k^i} \right\}$, where ${f_i}\left( t \right) \ge 0$ is the event-triggered condition to be designed later.
Here, it is assumed that ${t_0^i = 0}$.

Based on the above description, the following control mechanisms are considered:
\begin{equation}\label{xy1}
\begin{split}
{u_i}\left( t \right) = \left\{ {\begin{array}{*{20}{c}}
{0,}&{i \in \mathcal{R}}\\
{ - K{P_i}\left( {t_k^i} \right),}&{i \in \mathcal{F}}
\end{array}} \right.
\end{split}
\end{equation}
where ${P_i}\left( t \right) = \sum\nolimits_{j = 1}^m {{a_{ij}}\left( {{x_i}\left( t \right) - {x_j}\left( t \right)} \right)}  + \sum\nolimits_{j = m + 1}^N {{b_{ij}}\left( {{x_i}\left( t \right) - {x_j}\left( t \right)} \right)} $, and $K$ is the feedback gain matrix to be designed.
\begin{remark}
Noted that the control mechanism only need the state ${x_i}\left( {t_k^i} \right)$ rather than the real-time state of the agent $i$ within interval $t \in \left[ {t_k^i,t_{k + 1}^i} \right)$. Therefore, this reduces the number of data transmission during the operation of the system, improves the network efficiency and reduces energy consumption.
\end{remark}
We pursue the following event-trigger functions for determining the trigger instants in the analysis:

\begin{equation}\label{cfhs}
\begin{split}
{f_i}(t) = \int_{t_k^i}^t {{{\left\| {{e_i}(s)} \right\|}^2}ds}  - \int_{t_k^i}^t {\left( {k_i}{{\left\| {{P_i}( {t_k^i})} \right\|}^2} + \beta {e^{ - \sigma s}} \right)ds}  > 0,\\
i \in {\cal F}
\end{split}
\end{equation}

For the convenience of discussion,
define $ x\left( t \right) = {\left[ {{{ x}_1\left( t \right)}^T, \cdots ,{{ x}_m\left( t \right)}^T} \right]^T}$, $e\left( t \right) = {\left[ {{e_{{1}}}{{\left( t \right)}^T}, \cdots ,{e_{{m}}}{{\left( t \right)}^T}} \right]^T}$,
where
${e_{{i}}}\left( t \right) = {P_i}\left( {t_k^i} \right) - {P_i}\left( t \right)$ for $t \in \left[ {t_k^i,t_{k + 1}^i} \right)$.

By substituting protocol \eqref{xy1} into \eqref{xtfc1}, the closed-loop system can be summarized as
\begin{equation}
\begin{split}
{{\dot x}_i}\left( t \right)
=& A{x_i}\left( t \right) - BK{P_i}\left( t \right) - BK{e_i}\\
=& A{x_i}\left( t \right) - \left( {{{\cal L}_{i:}} \otimes BK} \right)x\left( {t} \right) - BK\sum\nolimits_{j = m + 1}^N {{b_{ij}}} {x_i}\left( {t} \right)\\
 &+ BK\sum\nolimits_{j = m + 1}^N {{b_{ij}}} {x_j}\left( {t} \right) - BK{e_i}
\end{split}
\end{equation}
Above systems can be given in a compact form of
\begin{equation}
\begin{split}
\dot x
= &\left[ {I_m} \otimes A - {\mathcal{L}_\mathcal{F}} \otimes \left( {BK} \right) - \sum\nolimits_{j = m + 1}^N diag\left\{ {{b_{1j}}, \ldots ,{b_{mj}}} \right\} \otimes \left( {BK} \right) \right]x\\
&+ \sum\nolimits_{j = m + 1}^N \left[ {diag\left\{ {{b_{1j}}, \ldots ,{b_{mj}}} \right\} \otimes \left( {BK} \right)}\left( {1_m} \otimes {x_j}\right)\right] - \left( {{I_m} \otimes BK} \right)e\left( t \right)\\
=& \left[ {{I_m} \otimes A - \left( {{\mathcal{L}_\mathcal{F}} + \sum\nolimits_{j = m + 1}^N {{B_{oj}}} } \right) \otimes \left( {BK} \right)} \right]x \\
&+ \sum\nolimits_{j = m + 1}^N {\left( {{B_{oj}} \otimes \left( {BK} \right)} \right)\left( {{1_m} \otimes {x_j}} \right)}- \left( {{I_m} \otimes BK} \right)e\left( t \right)\\
=& \left[ {{I_m} \otimes A - \rm M \otimes \left( {BK} \right)} \right]x - \left( {{I_m} \otimes BK} \right)e\left( t \right)\\
&+ \sum\nolimits_{j = m + 1}^N {\left( {{B_{oj}} \otimes \left( {BK} \right)} \right)\left( {{1_m} \otimes {x_j}} \right)}
\end{split}
\end{equation}
where ${B_{oj}} = diag\left\{ {{b_{1j}}, \ldots ,{b_{mj}}} \right\}$,
$\rm M = {\mathcal{L}_\mathcal{F}} + \sum\nolimits_{j = m + 1}^N {{B_{oj}}} $.

Define the state difference $\varepsilon \left( t \right)$ among leaders and followers as
\begin{equation}
\begin{split}
\varepsilon \left( t \right)=& \left( {\rm M \otimes {I_n}} \right)x - \sum\nolimits_{j = m + 1}^N {\left( {{B_{oj}} \otimes {I_n}} \right)\left( {1_m} \otimes {x_j} \right)} \\
=& \left( {\rm M \otimes {I_n}} \right)\left[ x - \sum\nolimits_{j = m + 1}^N \left( {\left( {{\rm M^{ - 1}}{B_{oj}}} \right) \otimes {I_n}} \right)\left( {{1_m} \otimes {x_j}} \right)  \right]
\end{split}
\end{equation}
The above equation holds because $\rm M$ is strictly diagonally dominant and ${\rm M^{- 1}}$ exists.

Denote a variable ${{\tilde x}_i} = {x_i} - \sum\nolimits_{j = m + 1}^N {{\chi _{ij}}{x_j}} $, where ${{\chi _{ij}}}$ is  the $i$-th element of ${\rm M^{ - 1}}{B_{oj}}{1_m}$.

Then
\begin{equation}\label{6}
\begin{split}
{{\dot {\tilde x}}_i} &= {{\dot x}_i} - \sum\nolimits_{j = m + 1}^N {{\chi _{ij}}{{\dot x}_j}} \\
 &= A{x_i} - BK{P_i}\left( {t_k^i} \right) - \sum\nolimits_{j = m + 1}^N {{\chi _{ij}}A{x_j}} \\
  &= A{x_i} - BK{P_i}\left( t \right) - \sum\nolimits_{j = m + 1}^N {{\chi _{ij}}A{x_j}}  - BK{e_i}\left( t \right)
\end{split}
\end{equation}
Furthermore, $\tilde x$ evolves according to
\begin{equation}\label{ds}
\begin{split}
\dot {\tilde x}=& \left( {{I_m} \otimes A} \right)x - \left[ {\left( {{\mathcal{L}_\mathcal{F}} + \sum\nolimits_{j = m + 1}^N {{B_{oj}}} } \right) \otimes BK} \right]x \\
&- \left( {{I_m} \otimes BK} \right)e(t) + \sum\nolimits_{j = m + 1}^N ({{B_{oj}} \otimes ({BK})})( {{1_m} \otimes {x_j}}) \\
&- \sum\nolimits_{j = m + 1}^N {\left( {\left( {{\rm M^{ - 1}}{B_{oj}}} \right) \otimes {I_n}} \right)\left( {{I_m} \otimes A} \right)\left( {{1_m} \otimes {x_j}} \right)} \\
= & ({{I_m} \otimes A})\left\{ x - \sum\nolimits_{j = m + 1}^N ({( {{\rm M^{ - 1}}{B_{oj}}}) \otimes {I_n}})( {{1_m} \otimes {x_j}}) \right\}\\
&+ \left( {\rm M \otimes BK} \right)\sum\nolimits_{j = m + 1}^N {\left( {\left( {{\rm M^{ - 1}}{B_{oj}}} \right) \otimes {I_n}} \right)\left( {{1_m} \otimes {x_j}} \right)}\\
&- \left( {{I_m} \otimes BK} \right)e\left( t \right)- \left( {\rm M \otimes BK} \right)x \\
=& \left( {{I_m} \otimes A - \rm M \otimes BK} \right)\left\{ x - \sum\nolimits_{j = m + 1}^N \left( {\left( {{\rm M^{ - 1}}{B_{oj}}} \right) \otimes {I_n}} \right)\times \left( {{1_m} \otimes {x_j}} \right)  \right\}\\
&- \left( {{I_m} \otimes BK} \right)e\left( t \right)\\
=& {\left( {{I_m} \otimes A - \rm M \otimes BK} \right)\tilde x - \left( {{I_m} \otimes BK} \right)e\left( t \right)}
\end{split}
\end{equation}
Before discussing the S-stabilizability of system under the event-triggered condition \eqref{cfhs}, we will first review the important property of $\rm M$.
\begin{remark}\label{bz1}
$\rm M$ is strictly diagonally dominant, which means that there exists a positive definite diagonal matrix $\Psi$ and a positive number $\eta$, such that $\Psi \rm M + {\rm M^T}\Psi  > \eta \Psi$, where $\Psi  = diag\left\{ {{\psi _1}, \ldots ,{\psi _N}} \right\}$.
\end{remark}
\begin{theorem}\label{dlm}
Under Assumptions \ref{gf} and \ref{js3}, consider multi-agent system \eqref{xtfc1} with protocol \eqref{xy1}.
Agent $i$ determines the triggering time sequence $\left\{ {t_k^i} \right\}_{k = 1}^\infty $ by \eqref{cfhs}.
If for each iSCC cell in weakly connected follower network $\mathcal{G}_\mathcal{F}$, there exists at least one vertex $i\left(i \in \mathcal{F} \right)$ such that ${b_{ij}} > 0\left(j \in \mathcal{R}\right)$, then $S-stabilizability$ of system \eqref{xtfc1} can be realized by choosing ${k_{\max }} < {\textstyle{{\varsigma {v_1}} \over {{\rho _1}\left\| H \right\| + {v_1}\varsigma }}}$, $\varsigma  = \eta {\lambda _{\min }}\left( \Psi  \right) - {v_1}$ .
\end{theorem}
\begin{proof}
Construct candidate Lyapunov functions \[V_1 = {{\tilde x}^T}\left( {\Psi  \otimes R} \right)\tilde x\]
Then the derivative of $V_1$ along the trajectories of system \eqref{ds} yields
\begin{equation}
\begin{split}
\dot V_1
=& 2{{\tilde x}^T}\left( {\Psi  \otimes R} \right)\left\{ {\left( {{I_m} \otimes A - \rm M \otimes BK} \right)\tilde x - \left( {{I_m} \otimes BK} \right)e} \right\}\\
=& 2{{\tilde x}^T}\left( {\Psi  \otimes \left( {RA} \right) - \left( {\Psi \rm M} \right) \otimes RBK} \right)\tilde x - 2{{\tilde x}^T}\left( {\Psi  \otimes RBK} \right)e\\
=& {{\tilde x}^T}\left\{ {\Psi  \otimes \left( {{A^T}R + RA} \right) - \left( {\Psi \rm M + {\rm M^T}\Psi } \right) \otimes RBK} \right\}\tilde x - 2{{\tilde x}^T}\left( {\Psi  \otimes RBK} \right)e
\end{split}
\end{equation}
According to the Remark \ref{bz1}, the above formula is equivalent to
\begin{equation}
\begin{split}
\dot V_1 \le & {{\tilde x}^T}\left\{ {\Psi  \otimes \left( {{A^T}R + RA} \right) - \eta \Psi  \otimes RBK} \right\}\tilde x - 2{{\tilde x}^T}\left( {\Psi  \otimes RBK} \right)e\\
=& {{\tilde x}^T}\left\{ {\Psi  \otimes \left( {{A^T}R + RA - \eta RBK} \right)} \right\}\tilde x - 2{{\tilde x}^T}\left( {\Psi  \otimes RBK} \right)e
\end{split}
\end{equation}
If $K=B^TR$, from the algebraic Riccati inequality, one can obtain that
\begin{equation}
\begin{split}
\dot V_1
&\le {{\tilde x}^T}\left\{ {\Psi  \otimes \left( { - \eta I} \right)} \right\}\tilde x - 2{{\tilde x}^T}\left( {\Psi  \otimes RBK} \right)e\\
 &\le  - \eta {\lambda _{\min }}\left( \Psi  \right){\left\| {\tilde x} \right\|^2} - 2{{\tilde x}^T}\left( {\Psi  \otimes RBK} \right)e
\end{split}
\end{equation}
Further, using the Lemma \ref{lemmays}, we have
\begin{equation}\label{yss}
\begin{split}
 - 2{{\tilde x}^T}\left( {\Psi  \otimes RBK} \right)e \le {v_1}{\left\| {\tilde x} \right\|^2} + {\textstyle{{{\rho _1}} \over {{v_1}}}}{\left\| e \right\|^2}
\end{split}
\end{equation}
where ${\rho _1} = {\lambda _{\max }}\left[ {{\Psi ^2} \otimes {{\left( {RBK} \right)}^2}} \right]$.
Based on \eqref{yss}, \eqref{vd} holds.
\begin{equation}\label{vd}
\begin{split}
\dot V \le  - \varsigma {\left\| {\tilde x} \right\|^2} + {\textstyle{{{\rho _1}} \over {{v_1}}}}{\left\| e \right\|^2}
\end{split}
\end{equation}
where $\varsigma  = \eta {\lambda _{\min }}\left( \Psi  \right) - {v_1}$.

Recalling the event-triggered condition \eqref{cfhs}, we obtain that
\begin{equation}\label{wcjf}
\begin{split}
\int_{t_k^i}^t {{{\left\| {{e_i}\left( s \right)} \right\|}^2}ds}  \le \frac{1}{{1 - {k_i}}}\int_{t_k^i}^t {\left( {{k_i}{{\left\| {{P_i}\left( s \right)} \right\|}^2} + \beta {e^{ - \sigma s}}} \right)ds}
\end{split}
\end{equation}
It can be further obtained from the above equation,
\scriptsize
\[\int_{{t_0}}^t {{{\left\| {e\left( s \right)} \right\|}^2}ds}  \le \frac{1}{{1 - {k_{\max }}}}\int_{{t_0}}^t {\left( {{k_{\max }}{{\left\| {\left( {H \otimes {I_n}} \right)\tilde x\left( s \right)} \right\|}^2} + m\beta {e^{ - \sigma s}}} \right)ds} \]
\normalsize
Integrating \eqref{vd} over $\left[ {{t_0},t} \right)$, yields
\begin{equation}\label{dsc}
\begin{split}
V_1\left( t \right) - V_1\left( {{t_0}} \right)
\le &  - \left( {\varsigma  - {\textstyle{{{\rho _1}{k_{\max }}\left\| H \right\|} \over {{v_1}\left( {1 - {k_{\max }}} \right)}}}} \right)\int_{{t_0}}^t {{{\left\| {\tilde x\left( s \right)} \right\|}^2}ds}\\
&+ \frac{{{\rho _1}N\beta }}{{{v_1}\sigma \left( {1 - {k_{\max }}} \right)}}\left( {{e^{ - \sigma {t_0}}} - {e^{ - \sigma t}}} \right)
\end{split}
\end{equation}
where ${k_{\max }} < {\textstyle{{\varsigma {v_1}} \over {{\rho _1}\left\| H \right\| + {v_1}\varsigma }}}$.

Because $V_1\left( t \right) > 0$, ${e^{ - \sigma t}} \ge 0$, which can be obtained from \eqref{dsc}
\[\int_{{t_0}}^t {{{\left\| {\tilde x\left( s \right)} \right\|}^2}ds}  \le \frac{{V_1\left( {{t_0}} \right){v_1}\sigma \left( {1 - {k_{\max }}} \right) + {\rho _1}N\beta {e^{ - \sigma {t_0}}}}}{{\varsigma {v_1}\sigma \left( {1 - {k_{\max }}} \right) - {\rho _1}{k_{\max }}\left\| \rm M \right\|}}\]
Therefore, $\int_{{t_0}}^t {{{\left\| {\tilde x\left( s \right)} \right\|}^2}ds} $ is bounded.
And \eqref{dsc} indicates $V_1\left( t \right)$ is bounded, hence ${\dot V}$ is bounded.
Furthermore, we can deduce that $\frac{d}{{d{t^2}}}\int_{{t_0}}^t {{{\left\| {\tilde x\left( s \right)} \right\|}^2}ds}$ is bounded.
Therefore, according to Barbalat's Lemma, we can get
\[\mathop {\lim }\limits_{t \to \infty } \frac{d}{{dt}}\int_{{t_0}}^t {{{\left\| {\tilde x\left( s \right)} \right\|}^2}ds}  = \mathop {\lim }\limits_{t \to \infty } {\left\| {\tilde x\left( s \right)} \right\|^2} = 0,\]
i.e.
\[\mathop {\lim }\limits_{t \to \infty } \left\{ {x - \sum\nolimits_{j = m + 1}^N {\left( {{\rm M^{ - 1}}{B_{oj}}{1_m}} \right) \otimes {x_j}} } \right\} = 0.\]

Next, we prove that the follower's state converges to the convex hull spanned by leaders.

Denote
${\mathcal{L}_{\mathcal{FR}}} = {{{\mathcal{L}_\mathcal{F}}} \mathord{\left/
 {\vphantom {{{\mathcal{L}_\mathcal{F}}} {\left( {n - m} \right)}}} \right.
 \kern-\nulldelimiterspace} {\left( {n - m} \right)}}$,
since $\left( {{\mathcal{L}_\mathcal{F}} \otimes {I_n}} \right)\left( {{1_m} \otimes {x_\mathcal{R}}} \right) = 0$,
hence
\begin{equation}
\begin{split}
\varepsilon \left( t \right)=& \left( {\rm M \otimes {I_n}} \right)x - \sum\nolimits_{j = m + 1}^N {\left( {{B_{oj}} \otimes {I_n}} \right)\left( {{1_m} \otimes {x_j}} \right)} \\
=& \left( {\rm M \otimes {I_n}} \right)x - \sum\nolimits_{j = m + 1}^N \left( {\left( {{\mathcal{L}_{\mathcal{FR}}} + {B_{oj}}} \right) \otimes {I_n}} \right)\times \left( {{1_m} \otimes {x_j}} \right) \\
=& \left( {\rm M \otimes {I_n}} \right)\left\{ x - \sum\nolimits_{j = m + 1}^N \left( {{\rm M^{ - 1}}\left( {{\mathcal{L}_{\mathcal{FR}}} + {B_{oj}}} \right){1_m}} \right)\otimes {x_j}  \right\}
\end{split}
\end{equation}
Note that
\begin{equation*}
\begin{split}
&\sum\nolimits_{j = m + 1}^N {{\rm M^{ - 1}}\left( {{\mathcal{L}_{\mathcal{FR}}} + {B_{oj}}} \right){1_m}}\\
=& \sum\nolimits_{j = m + 1}^N {{\rm M^{ - 1}}\left( {{{{\mathcal{L}_\mathcal{F}}} \mathord{\left/{\vphantom {{{L_\mathcal{F}}} {\left( {n - m} \right)}}} \right.
 \kern-\nulldelimiterspace} {\left( {n - m} \right)}} + {B_{oj}}} \right){1_m}} \\
= & \sum\nolimits_{j = m + 1}^N {{\rm M^{ - 1}}\rm M{1_m}}  = {1_m}
 \end{split}
\end{equation*}
Consequently,
\begin{equation}
\begin{split}
&\sum\nolimits_{j = m + 1}^N {\left( {{\rm M^{ - 1}}\left( {{\mathcal{L}_{\mathcal{FR}}} + {B_{oj}}} \right){1_m}} \right) \otimes {x_j}}\\
= &\sum\nolimits_{j = m + 1}^N {\left( {{\rm M^{ - 1}}{B_{oj}}{1_m}} \right) \otimes {x_j}}
\end{split}
\end{equation}
is the column vector of a convex combination of points in $X = \left\{ {{x_{m + 1}}, \ldots ,{x_N}} \right\}$.
Thus, it is concluded that system \eqref{xtfc1} achieves S-stabilizability.
\end{proof}
\begin{remark}
Obviously, $\mathop {\lim }\limits_{t \to \infty } {x_j}\left( t \right) = 0$ for $j \in \left\{ {m + 1, \ldots ,N} \right\}$ if $A$ is Hurwitz matrix, then $\mathop {\lim }\limits_{t \to \infty } {x_i}\left( t \right) = 0$, where $i \in \left\{ {1, \ldots ,m} \right\}$. That is, system \eqref{xtfc1} realizes stabilizability.

\end{remark}
\begin{remark}
The event-trigger mechanism is distributed, because the control protocol of each agent only depends on the state of itself and its neighbors, without any prior knowledge of global parameters. In addition, from the event-trigger condition \eqref{cfhs}, each agent does not need to monitor the state of neighbors continuously, hence this greatly reduces the frequency of driving updates and communication among agents.
\end{remark}
\begin{theorem}\label{dl6}
Under the conditions of \eqref{cfhs}, the system \eqref{xtfc1} does not exhibit Zeno behavior.
The interval between any two consecutive event-trigger instants of the system is not less than
\[{\textstyle{1 \over {\left\| A \right\|}}}\ln \left( {1 + {\textstyle{{\left\| A \right\|} \over {{h_i}}}}{{\left( {\beta {e^{ - \sigma \left( {t_k^i + \tau _k^i} \right)}}} \right)}^{{\textstyle{1 \over 2}}}}} \right)\]
\end{theorem}
\begin{proof}
According to the definition of ${e_i}\left( t \right)$, we can obtain that
\begin{equation}
\begin{split}
{{\dot e}_i}\left( t \right)=&  - {{\dot P}_i}\left( t \right)\\
=&  - \sum\nolimits_{j = 1}^m {{a_{ij}}\left( {A{x_i}\left( t \right) - BK{P_i}\left( {t_k^i} \right)} \right)}+ \sum\nolimits_{j = 1}^m {{a_{ij}}\left( {A{x_j}\left( t \right) - BK{P_j}\left( {t_{k'}^j} \right)} \right)} \\
&- \sum\nolimits_{j = m + 1}^N {{b_{ij}}\left( {A{x_i}\left( t \right) - BK{P_i}\left( {t_k^i} \right)} \right)}+ \sum\nolimits_{j = m + 1}^N {{b_{ij}}\left( {A{x_j}\left( t \right) - BK{P_j}\left( {t_{k'}^j} \right)} \right)}\\
=&  - \sum\nolimits_{j = 1}^m {{a_{ij}}\left( {A{x_i}\left( t \right) - A{x_j}\left( t \right)} \right)}- \sum\nolimits_{j = m + 1}^N {b_{ij}}\left( {A{x_i}\left( t \right)- A{x_j}\left( t \right)} \right)\\
&+ BK\sum\nolimits_{j = 1}^m {{a_{ij}}\left( {{P_i}\left( {t_k^i} \right) - {P_j}\left( {t_{k'}^j} \right)} \right)}+ BK\sum\nolimits_{j = m + 1}^N {{b_{ij}}\left( {{P_i}\left( {t_k^i} \right) - {P_j}\left( {t_{k'}^j} \right)} \right)} \\
=& - A{P_i}\left( t \right) + BK\left[ \sum\nolimits_{j = 1}^m {{a_{ij}}\left( {{P_i}\left( {t_k^i} \right) - {P_j}\left( {t_{k'}^j} \right)} \right)} + \sum\nolimits_{j = m + 1}^N {{b_{ij}}\left( {{P_i}\left( {t_k^i} \right) - {P_j}\left( {t_{k'}^j} \right)} \right)}  \right]\\
=& BK\sum\nolimits_{j = 1}^m {{a_{ij}}\left( {{P_i}\left( {t_k^i} \right) - {P_j}\left( {t_{k'}^j} \right)} \right)} + BK\sum\nolimits_{j = m + 1}^N {{b_{ij}}\left( {{P_i}\left( {t_k^i} \right) - {P_j}\left( {t_{k'}^j} \right)} \right)} \\
&+ A\left( {{P_i}\left( {t_k^i} \right) - {P_i}\left( t \right) - {P_i}\left( {t_k^i} \right)} \right)\\
=& BK\sum\nolimits_{j = 1}^m {{a_{ij}}\left( {{P_i}\left( {t_k^i} \right) - {P_j}\left( {t_{k'}^j} \right)} \right)}+ BK\sum\nolimits_{j = m + 1}^N {{b_{ij}}\left( {{P_i}\left( {t_k^i} \right) - {P_j}\left( {t_{k'}^j} \right)} \right)}\\
&+ A{e_i}\left( t \right) - A{P_i}\left( {t_k^i} \right)
\end{split}
\end{equation}

For $t \in \left[ {t_k^i,t_{k + 1}^i} \right)$, it can be derived that
\begin{equation}\label{18}
\begin{split}
\frac{{d\left\| {{e_i}\left( t \right)} \right\|}}{{dt}} \le \left\| {{{\dot e}_i}\left( t \right)} \right\| \le \left\| A \right\|\left\| {{e_i}\left( t \right)} \right\| + {\Gamma _i}
\end{split}
\end{equation}
where
\begin{equation*}
\begin{split}
{\Gamma _i} &= \left\|- A{P_i}\left( {t_k^i} \right) + BK\left[ \sum\nolimits_{j = 1}^m {{a_{ij}}\left( {{P_i}\left( {t_k^i} \right) - {P_j}\left( {t_{k'}^j} \right)} \right)}\right.\right.\\
 & \left.\left.+ \sum\nolimits_{j = m + 1}^N {{b_{ij}}\left( {{P_i}\left( {t_k^i} \right) - {P_j}\left( {t_{k'}^j} \right)} \right)} \right] \right\|,
\end{split}
\end{equation*}
From Theorem \ref{dlm}, one can obtain that ${x_i}(t)$ is bounded, corresponding ${P_i}\left( {t_{k}^i} \right)$ is bounded, so there exists a constant ${h_i}$ that satisfies ${\Gamma _i} \le {h_i}$.
According to \eqref{18}, it can be induced that
\begin{equation}\label{}
\begin{split}
\frac{{d\left\| {{e_i}\left( t \right)} \right\|}}{{dt}} \le \left\| A \right\|\left\| {{e_i}\left( t \right)} \right\| + {h_i}
\end{split}
\end{equation}
Consider a nonnegative function, $\Xi :\left[ {0,\infty } \right) \to {\mathbb{R}_{ \ge 0}}$, which satisfies
\begin{equation}\label{19}
\begin{split}
\dot \Xi  = \left\| A \right\|\Xi  + {h_i}, \Xi \left( 0 \right) = \left\| {{e_i}\left( {t_k^i} \right)} \right\| = 0
\end{split}
\end{equation}
where $\Xi \left( t \right) = {\textstyle{{{h_i}} \over {\left\| A \right\|}}}\left( {{e^{\left\| A \right\|t}} - 1} \right)$ is the solution of \eqref{19}.
According to Lemma \ref{lemma3}, we can easily get $\left\| {{e_i}\left( t \right)} \right\| \le \Xi \left( {t - t_k^i} \right)$.
From event trigger function \eqref{cfhs}, if
\begin{equation}\label{20}
\begin{split}
{\left\| {{e_i}\left( t \right)} \right\|^2} \le \beta {e^{ - \sigma t}}
\end{split}
\end{equation}
then ${f_i}\left( t \right) \le 0$.
Therefore, we can know that the lower bound of the event-trigger interval of agent $i$ can be lower bounded by the evolution time for ${\Xi ^2}\left( {t - t_k^i} \right)$ to evolve from $0$ to $\beta {e^{ - \sigma t}}$,
that is the lowest bound of $\tau _k^i$ can be obtained by \eqref{52}:
\begin{equation}\label{52}
\begin{split}
\beta {e^{ - \sigma \left( {t_k^i + \tau _k^i} \right)}} = {\textstyle{{h_i^2} \over {{{\left\| A \right\|}^2}}}}{\left( {{e^{\left\| A \right\|\tau _k^i}} - 1} \right)^2}
\end{split}
\end{equation}
The above equation is equivalent to
\[\tau _k^i = {\textstyle{1 \over {\left\| A \right\|}}}\ln \left( {1 + {\textstyle{{\left\| A \right\|} \over {{h_i}}}}{{\left( {\beta {e^{ - \sigma \left( {t_k^i + \tau _k^i} \right)}}} \right)}^{{\textstyle{1 \over 2}}}}} \right)\]

The following conclusion can be proved by contradiction.
It is assumed that the Zeno behavior occurs, which means that there exists a positive constant ${t^*}$ such that $\mathop {\lim }\limits_{t \to \infty } t_k^i = {t^ * }$. Let ${\varepsilon _0} = \frac{1}{2}\tau _k^i $.
There exists a positive integer ${N_0}$ such that ${t^*} - {\varepsilon _0} \le t_k^i \le {t^*}$ for ${\varepsilon _0} > 0$ by the definition of sequence limit, where $k \ge {N_0}$.
Therefore, ${t^*} + {\varepsilon _0} \le t_k^i + 2{\varepsilon _0} \le {t_{k + 1}^i}$ holds when $k \ge {N_0}$.
This contradicts with ${t^ * } \ge {t_{k + 1}^i}$ for $k \ge {N_0}$. Thus, Zeno behavior is strictly excluded.

The proof is completed.
\end{proof}
\begin{remark}
If $k_i=0$, we call \eqref{cfhs} as a state-independent event-triggered condition. \eqref{cfhs} is named as a state-dependent event-triggered condition if $\beta  = 0$.
These two event conditions are feasible to obtain stabilizability of system \eqref{ds}.
Thus, \eqref{cfhs} can be named as hybrid trigger condition, which is universal.
\end{remark}
\subsection{Dynamic Event-Triggered Control (DETC)}
In this section, in order to improve the effect of event-triggered mechanism, we introduce dynamic variable ${\varphi _i}$ to consider the stabilizability of the system:
\begin{equation}
\begin{split}\label{zjbl}
{{\dot \varphi }_i}\left( t \right)= - {\mu _i}{\varphi _i}\left( t \right) + {\xi _i} &( {{k_i}{{\left\| {{P_i}({t_k^i})} \right\|}^2} + \beta{e^{ - \sigma t}}-{{\left\| {{e_i}(t)} \right\|}^2}}),\\
& {\Theta _i} > {\varphi _i}\left( 0 \right) > 0,{\mu _i} > 0,{\xi _i} > 0
\end{split}
\end{equation}
\begin{theorem}\label{d3}
Suppose the communication topology is a directed graph and weakly connected.
Given parameter ${k_{\max }} < {\textstyle{{\varsigma {v_1}} \over {{\rho _1}\left\| H \right\| + {v_1}\varsigma }}}$ and the first event trigger time $t_1^i = 0$, the trigger time of agent $i$ is determined by the following trigger function:
\begin{equation}
\begin{split}\label{sjcfhs2}
t_{k + 1}^i = \mathop {\max }\limits_{r \ge t_k^i} & \left\{ r:{\varphi _i}\left( t \right) \ge {\theta _i}( {{{\left\| {{e_i}} \right\|}^2} - {k_i}{{\left\| {{P_i}\left( {t_k^i} \right)} \right\|}^2} - \beta {e^{ - \sigma t}}} ), \forall t \in \left[ {t_k^i,r} \right] \right\}
\end{split}
\end{equation}
If for each iSCC cell of the follower network $\mathcal{G}_\mathcal{F}$, there exists at least one vertex $i$ such that ${b_{ij}} > 0(i \in \mathcal{F} , j \in \mathcal{R})$,
then the S-stabilizability of the multi-agent system \eqref{ds} can be realized under the event triggered protocol \eqref{sjcfhs2}, and there is no Zeno behavior in the closed-loop system.
\end{theorem}
\begin{proof}
According to the function \eqref{zjbl} and trigger condition \eqref{sjcfhs2}, we have
\begin{equation}
\begin{split}
{{\dot \varphi }_i}\left( t \right) \ge  - {\mu _i}{\varphi _i}\left( t \right) - {\textstyle{{{\xi _i}} \over {{\theta _i}}}}{\varphi _i}\left( t \right) =  - \left( {{\mu _i} + {\textstyle{{{\xi _i}} \over {{\theta _i}}}}} \right){\varphi _i}\left( t \right)
\end{split}
\end{equation}
so
\begin{equation}\label{gs15}
\begin{split}
{\varphi _i}\left( t \right) > {\varphi _i}\left( 0 \right){e^{ - \left( {{\mu _i} + {\textstyle{{{\xi _i}} \over {{\theta _i}}}}} \right)t}} > 0.
\end{split}
\end{equation}
Constructing Lyapunov candidate function
\begin{equation}
\begin{split}
V = {V_1} + {V_2}
\end{split}
\end{equation}
where
\begin{equation*}
\begin{split}
{V_2} = \sum\nolimits_{i = 1}^N {{\varphi _i}\left( t \right)}
\end{split}
\end{equation*}
Then
\begin{equation*}
\begin{split}
{{\dot V}_2}= &\sum\nolimits_{i = 1}^N {\xi _i}\left( {k_i}{{\left\| {{P_i}\left( {t_k^i} \right)} \right\|}^2}+ \beta {e^{ - \sigma t}} - {{\left\| {{e_i}\left( t \right)} \right\|}^2} \right) -\sum\nolimits_{i = 1}^N {{\mu _i}{\varphi _i}\left( t \right)}\\
 \le & - \sum\nolimits_{i = 1}^N {{\mu _i}{\varphi _i}\left( t \right)}  - {\xi _{\max }}\left( {1 - {k_{\max }}} \right){\left\| \rm M \right\|^2}{\left\| e \right\|^2} \\
&+{\xi _{\max }}{k_{\max }}{\left\| \rm M \right\|^2}{\left\| {\tilde x} \right\|^2} + m\beta {e^{ - \sigma t}}
\end{split}
\end{equation*}
The derivative of $V$ along the trajectory \eqref{ds} is
\begin{equation}\label{40}
\begin{split}
\dot V \le &  - \varsigma {\left\| {\tilde x} \right\|^2}{\rm{ + }}\frac{{{\rho _1}}}{{{v_1}}}{\left\| e \right\|^2} - \sum\nolimits_{i = 1}^N {{\mu _i}{\varphi _i}\left( t \right)}+ m\beta {e^{ - \sigma t}} \\
& - {\xi _{\max }}\left( {1 - {k_{\max }}{{\left\| \rm M \right\|}^2}} \right){\left\| e \right\|^2}+ {\xi _{\max }}{k_{\max }}{\left\| \rm M \right\|^2}{\left\| {\tilde x} \right\|^2} \\
\le &  - \vartheta {\left\| {\tilde x} \right\|^2} - \sum\nolimits_{i = 1}^N {{\mu _i}{\varphi _i}\left( t \right)}  + m\beta {e^{ - \sigma t}}
\end{split}
\end{equation}
where
$\vartheta  = \varsigma  - {\xi _{\max }}{k_{\max }}{\left\| \rm M \right\|^2}$,
${\xi _{\max }} = {\textstyle{{{\rho _1}} \over {{v_1}\left( {1 - {k_{\max }}{{\left\| H \right\|}^2}} \right)}}}$.

Let
${k_w} = \min \left\{ {{\textstyle{\vartheta  \over {{\lambda _{\max }}\left( {\Psi  \otimes P} \right)}}},{\mu _i}} \right\} > 0$,
then
\[\dot V \le  - {k_w}V + m\beta {e^{ - \sigma t}}\]
According to the comparison principle, we have
$0 \le {V_1} \le \psi \left( t \right)$,
where $\dot \psi \left( t \right) =  - {k_w}\psi \left( t \right) + m\beta {e^{ - \sigma t}}$,
$\psi \left( 0 \right) = {V_1}\left( 0 \right)$.
Therefore, we can further obtain that
\begin{equation}\label{g9}
\begin{split}
\psi \left( t \right) = \left\{ {\begin{array}{*{20}{c}}
{{e^{ - {k_w}t}}\psi \left( 0 \right) + m\beta t{e^{ - {k_w}t}},}&{{k_w} = \sigma }\\
{{e^{ - {k_w}t}}\psi \left( 0 \right) + \frac{{m\beta }}{{\left( {{k_w} - \sigma } \right)}}\left( {{e^{ - \sigma t}} - {e^{ - {k_w}t}}} \right),}&{{k_w} \ne \sigma }
\end{array}} \right.
\end{split}
\end{equation}
Obviously, when $t \to \infty $, $\psi \left( t \right) \to 0$ holds.
Therefore, we can deduce that $V\left( t \right) \to 0$ when $t \to \infty $, i.e.
$\mathop {\lim }\limits_{t \to \infty } {{\tilde x}_i}\left( t \right) = 0$.
Moreover, system \eqref{ds} is convergent exponentially, so the S-stabilizability is solved.

Next, we prove that the event-trigger interval among agents has a strict lower bound of positive time to exclude Zeno behavior. According to \eqref{18}, one get
\begin{equation}\label{30}
\begin{split}
\frac{{d{{\left\| {{e_i}\left( t \right)} \right\|}^2}}}{{dt}}&= 2\left\| {{e_i}\left( t \right)} \right\|\frac{{d\left\| {{e_i}\left( t \right)} \right\|}}{{dt}}\\
&= 2\left\| A \right\|{\left\| {{e_i}\left( t \right)} \right\|^2}{\rm{ + }}2\left\| {{e_i}\left( t \right)} \right\|{\Gamma _i}
\end{split}
\end{equation}
In addition,
\begin{equation}\label{32}
\begin{split}
\frac{{d{{\left\| {{e_i}\left( t \right)} \right\|}^2}}}{{dt}} \le  \left( {2\left\| A \right\| + 1} \right){\left\| {{e_i}\left( t \right)} \right\|^2}  + \Gamma _i^2
\end{split}
\end{equation}
holds, because of $2\left\| {{e_i}\left( t \right)} \right\|{\Gamma _i} \le {\left\| {{e_i}\left( t \right)} \right\|^2} + \Gamma _i^2$.
With $\left\| {{e_i}\left( {t_k^i} \right)} \right\| = 0$, and ${\Gamma _i} \le {h_i}$,
integrating \eqref{32} from ${t_k^i}$ to $t$, one can obtain
that
\begin{equation}\label{33}
\begin{split}
{\left\| {{e_i}\left( t \right)} \right\|^2} \le \int_{t_k^i}^t {{e^{\sigma \left( {t - s} \right)}}h_i^2ds}  = \frac{{h_i^2}}{\sigma }\left( {{e^{\sigma \left( {t - t_k^i} \right)}} - 1} \right)
\end{split}
\end{equation}
Based on the event-triggered condition \eqref{sjcfhs2}, the following equation can be obtained
\begin{equation}\label{34}
\begin{split}
{\left\| {{e_i}\left( t \right)} \right\|^2} > \frac{{{\varphi _i}}}{{{\theta _i}}} + {k_i}{\left\| {{P_i}\left( {t_k^i} \right)} \right\|^2} + \beta {e^{ - \sigma t}} \ge \frac{{{\Theta _i}}}{{{\theta _i}}}
\end{split}
\end{equation}
Combined with the formula \eqref{33}\eqref{34},
it can be deduced that the low bound of event-trigger interval of agent $i$ is
${\tau _i} = \frac{1}{\sigma }\ln \left( {1 + \frac{{\sigma {\Theta _i}}}{{{\theta _i}h_i^2}}} \right) > 0$.
Therefore, no Zeno behavior will exhibit.
\end{proof}
\begin{remark}
Obviously, when ${\theta _i}$ is infinite, the static event-triggered condition \eqref{cfhs} can be regarded as a limit case of the dynamic trigger condition \eqref{sjcfhs2}.
\end{remark}
\section{Simulation}
In this section, for verifying the accuracy of theoretical results, we perform a series of simulation experiments on the stabilizability of systems.
Consider a group of general linear multi-agent systems with
\[A_1 = \left[ {\begin{array}{*{20}{c}}
{ 1}&1\\
2&{ - 3}
\end{array}} \right],A_2 = \left[ {\begin{array}{*{20}{c}}
{ - 1}&1\\
2&{ - 3}
\end{array}} \right],B = \left[ {\begin{array}{*{20}{c}}
{ - 1}\\
0
\end{array}} \right].\]
\begin{figure}[!h]
  \centering
  \includegraphics[width=2in]{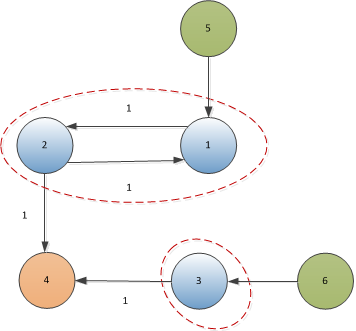}\\
  \caption{Communication topology graph}
\end{figure}
By solving Riccati inequality by MATLAB, the feedback gain matrices
$K_1 = \left[ {\begin{array}{*{20}{c}}
{ - 0.0254}&{ - 0.0012}
\end{array}} \right]$,
$K_2 = \left[ {\begin{array}{*{20}{c}}
{ - 0.1450}&{ - 0.0425}
\end{array}} \right]$ can be obtained.
Assume system \eqref{xtfc1} consists of 6 agents with
${x_{{1_0}}} = \left[ {\begin{array}{*{20}{c}}
{ - 2}&{4.3}
\end{array}} \right]$,
${x_{{2_0}}} = \left[ {\begin{array}{*{20}{c}}
7&{ - 4.5}
\end{array}} \right]$,
${x_{{3_0}}} = \left[ {\begin{array}{*{20}{c}}
{ - 4}&3
\end{array}} \right]$,
${x_{{4_0}}} = \left[ {\begin{array}{*{20}{c}}
8&2
\end{array}} \right]$,
${x_{{,5_0}}} = \left[ {\begin{array}{*{20}{c}}
2&2
\end{array}} \right]$,
${x_{{6_0}}} = \left[ {\begin{array}{*{20}{c}}
2&1
\end{array}} \right]$,
 where $\mathcal{F} = \left\{ {1, \ldots ,4} \right\}$ is follower set,
$\mathcal{R} = \left\{ {5, 6} \right\}$ is leader set.
The communication topology is described in Fig.1. According to the previous analysis,
let vertices $1$ and $3$ receive the leader's information.
\begin{figure}[!h]
  \centering
  \includegraphics[width=2.5in]{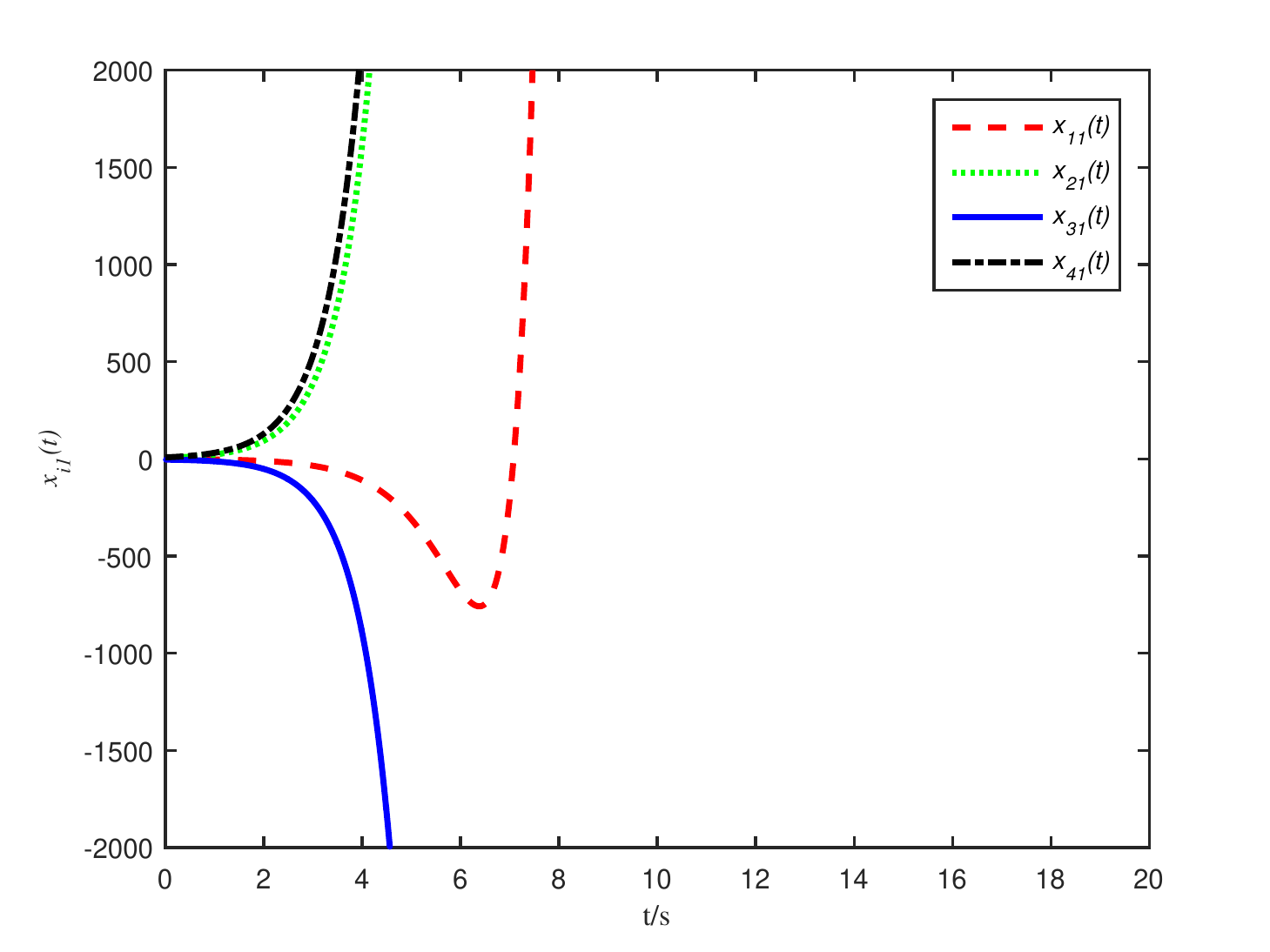}\\
  \caption{States $x_{i1}$ of followers with $A_1$ and $B$}
\end{figure}
\begin{figure}[!h]
  \centering
  \includegraphics[width=2.5in]{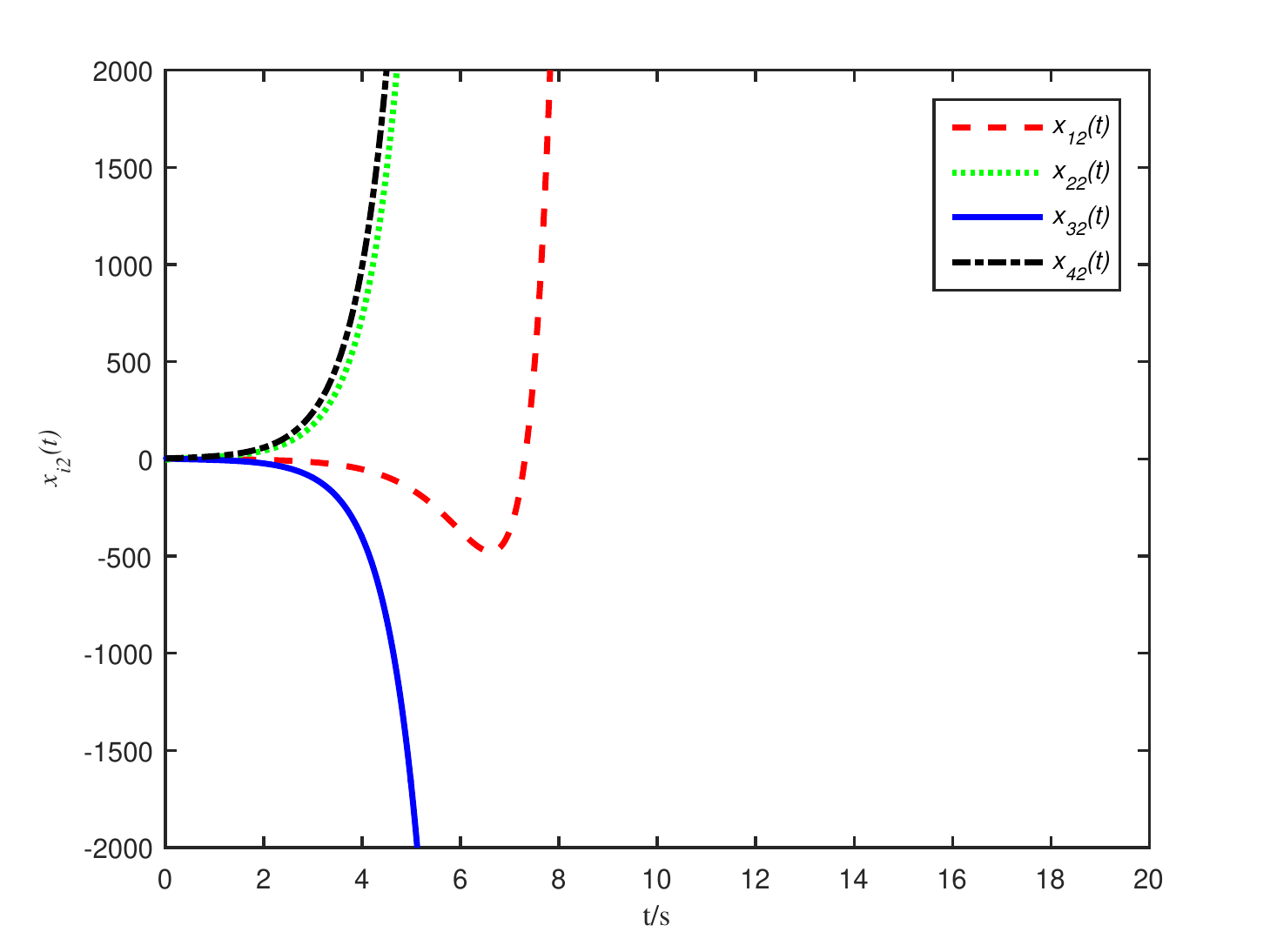}\\
  \caption{States $x_{i2}$ of followers with $A_1$ and $B$}
\end{figure}
Fig.2 and Fig.3 show the state trajectory of the followers when the system matrix is $A_1$. It can be found that the system is divergent in this case.
The corresponding renderings when the system matrix is $A_2$ are shown in Fig.4 - Fig.8.
\begin{figure}[!h]
  \centering
  \includegraphics[width=2.5in]{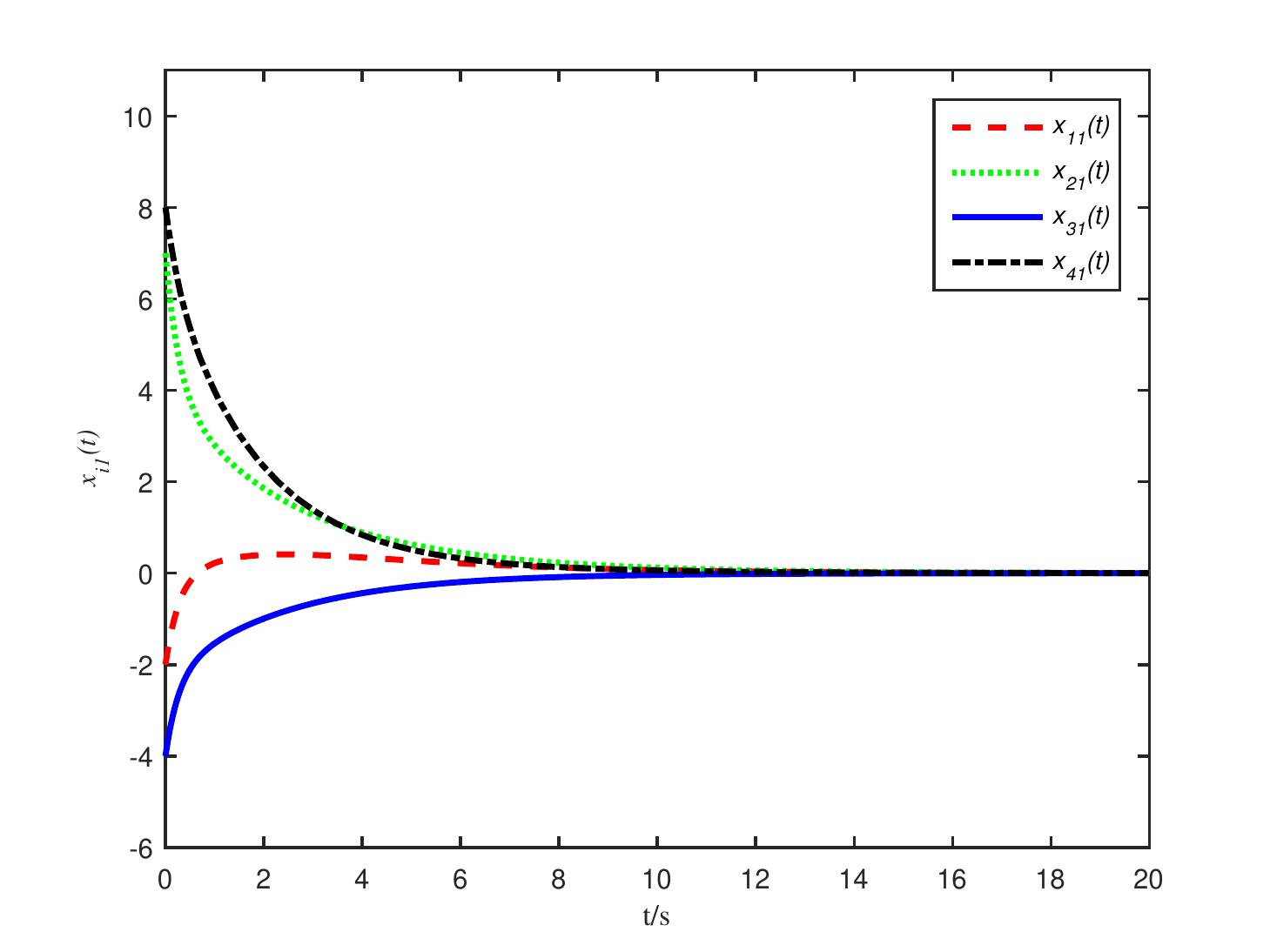}\\
  \caption{States $x_{i1}$ of followers with $A_2$ and $B$}
\end{figure}
\begin{figure}[!h]
  \centering
  \includegraphics[width=2.5in]{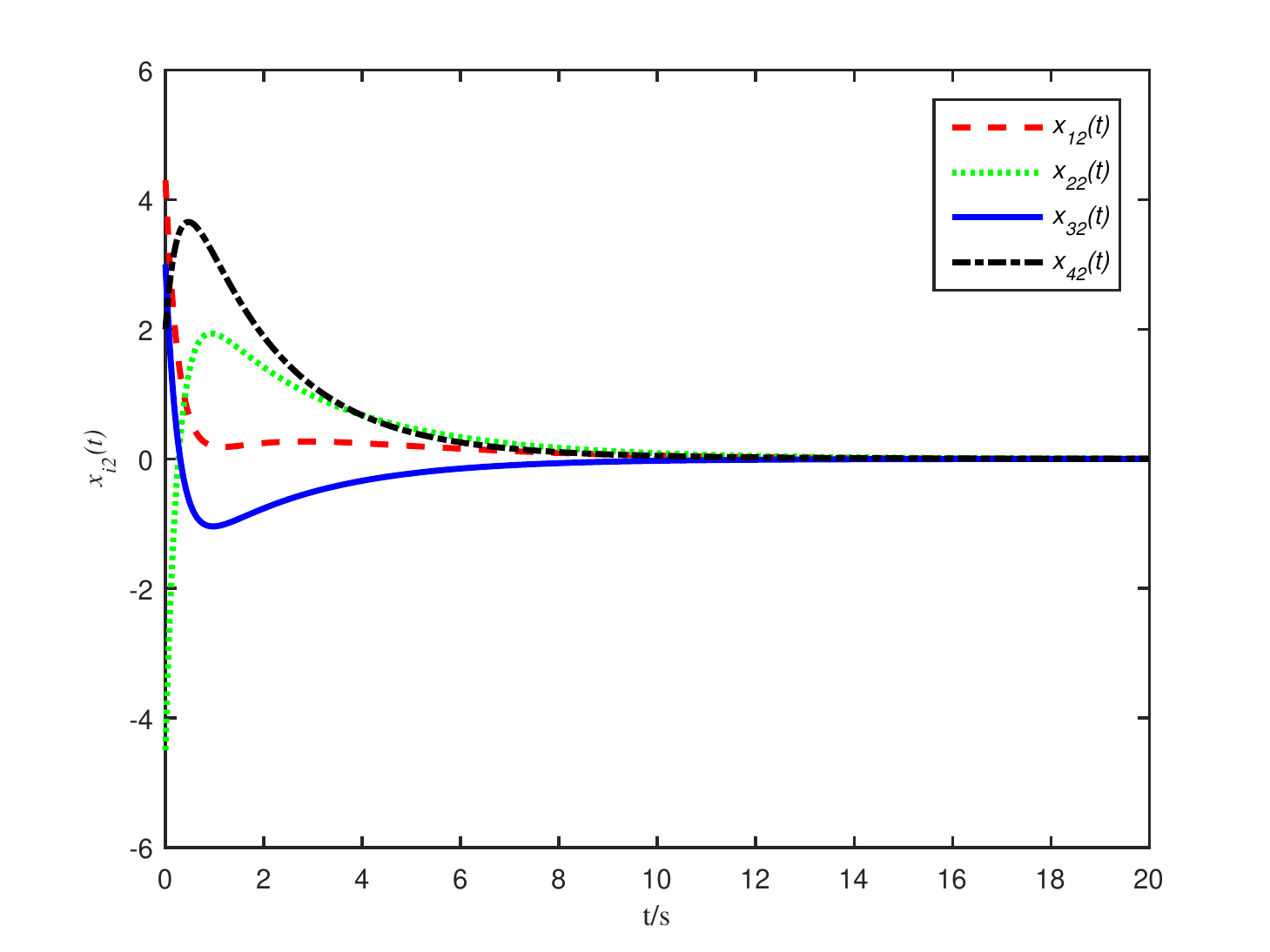}\\
  \caption{States $x_{i2}$ of followers with $A_2$ and $B$}
\end{figure}
\begin{figure}[!h]
  \centering
  \includegraphics[width=2.5in]{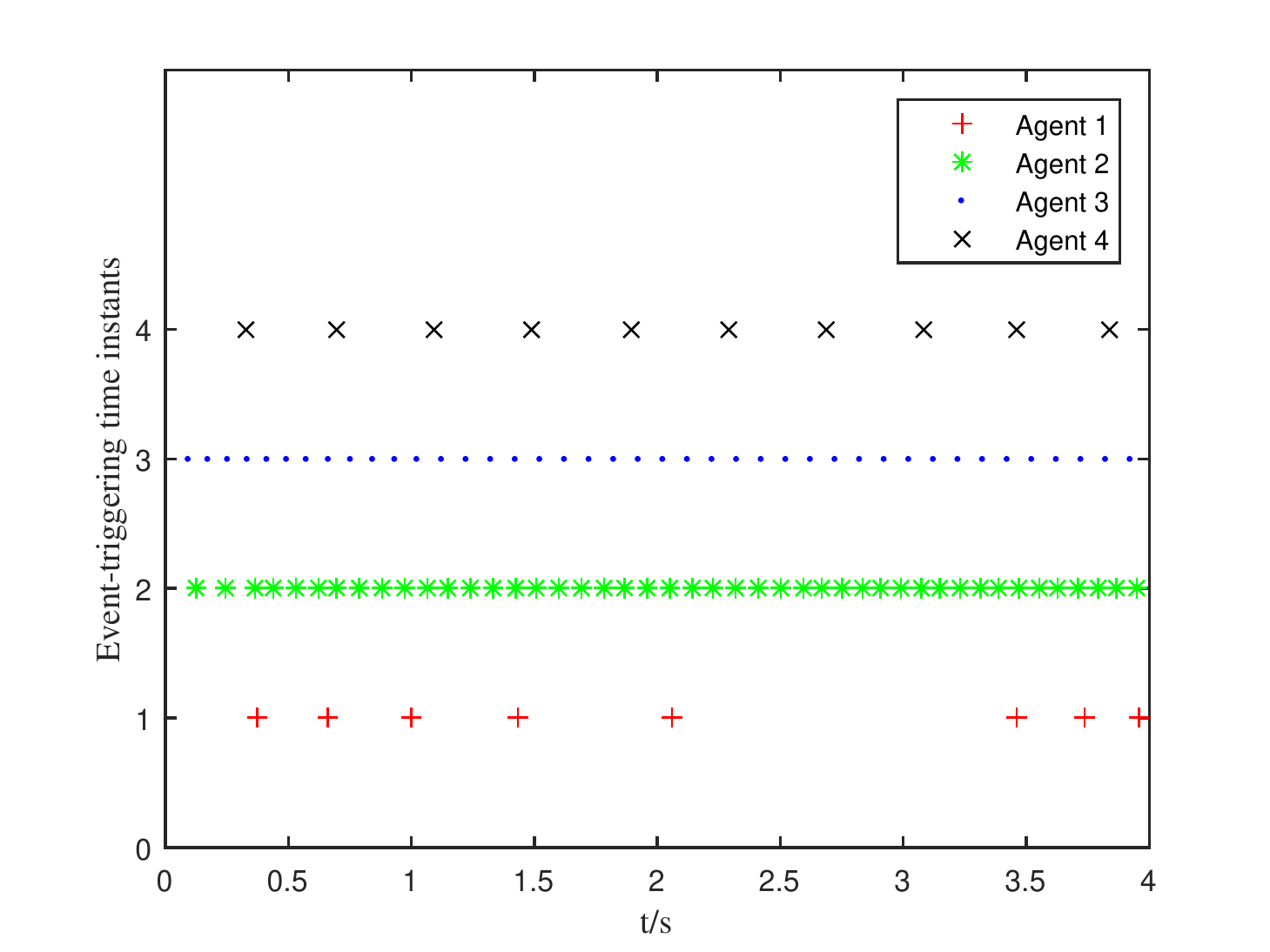}\\
  \caption{Triggering times of followers in \eqref{cfhs}}
\end{figure}
\begin{figure}[!h]
  \centering
  \includegraphics[width=2.5in]{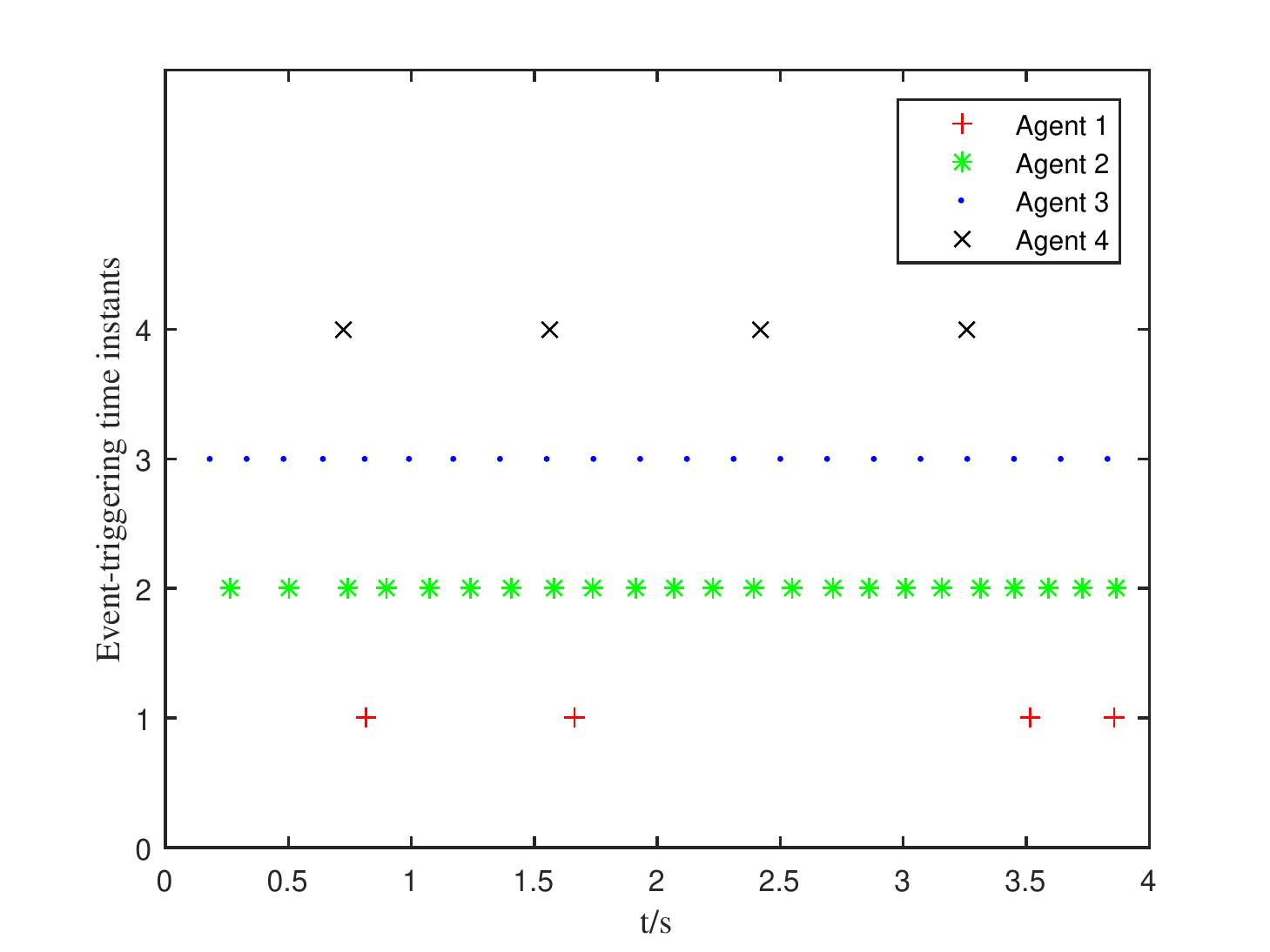}\\
  \caption{Triggering times of followers in \eqref{zjbl}}
\end{figure}
\begin{table}[htbp]
\centering
\caption{\text{The number of triggers for two event-triggered Schemes}}
 %\caption{\label{tab:test}}
 \begin{tabular}{ccccc}
  \toprule
 Type of event & Agent 1 & Agent 2 & Agent 3 & Agent 4\\
  \midrule

SETC & 524 & 638 & 401 & 245  \\
  \midrule

DETC  & 395 & 418 & 236 & 203  \\
  \bottomrule
 \end{tabular}
\end{table}
The evolutions of state ${x_i},i = 1,2,3,4,$ with $A_2$ and $B$ are shown in Fig.4 and Fig.5.
Fig.6 and Fig.7 show the event instants corresponding to two event-triggered conditions in \eqref{cfhs} and \eqref{zjbl}, respectively.
Obviously, there is no Zeno behavior.
Furthermore, in order to compare the two event-triggered law, we present the number of event triggering in TABLE I, respectively.
It can be seen that the dynamic event-triggered law guarantees a larger event interval than the static trigger law.

\section{Conclusion}
In this note, we have studied how to select control vertices to achieve the S-stabilizability of general linear multi-agent systems under event-triggered conditions by graph partition.
A new class of event-triggered protocols has been proposed for solving the S-stabilizability on directed topology.
Under this protocol, static and dynamic event-triggered conditions were proposed, respectively, and some sufficient conditions to ensure the S-stabilizability of the system were derived.
And we confirmed that stabilizability can be realized if $A$ is Hurwitz matrix.
In addition, it has been proved that the proposed static event-triggered condition is a limit case of dynamic event-trigger condition.
Future work will focus on solving the stabilizability of systems under switching topology.

\end{document}